\documentclass{llncs}
\usepackage{llncsdoc}
\usepackage{algpseudocode}
\usepackage{algorithm}
\usepackage{graphicx}

\begin{document}

\title{Weighted Sampling Without Replacement from Data Streams}

\author{Vladimir Braverman\thanks{Johns Hopkins University, Department of Computer Science. Email: vova@cs.jhu.edu. This material is based upon work supported in part by the National Science Foundation under Grant No. 1447639, the Google Faculty Award and DARPA grant N660001-1-2-4014. Its contents are solely the responsibility of the authors and do not represent the official view of DARPA or the Department of Defense.} Rafail Ostrovsky\thanks{University of California Los Angeles, Department of Computer Science and Department
of Mathematics, Email: rafail@cs.ucla.edu. } Gregory Vorsanger\thanks{Johns Hopkins University, Department of Computer Science. gregvorsanger@jhu.edu. This material is based upon work supported in part by Raytheon BBN Technologies.}}

\institute{}

\maketitle

\begin{abstract}

Weighted sampling without replacement has proved to be a very important tool in designing new algorithms.
Efraimidis and Spirakis (IPL 2006) presented an algorithm for weighted sampling
without replacement from data streams.
Their algorithm works under the assumption of precise computations over the interval $[0,1]$.
Cohen and Kaplan (VLDB 2008) used similar methods for their bottom-k
sketches.

Efraimidis and Spirakis ask as an open question whether using finite precision arithmetic impacts the accuracy of their algorithm. In this paper we show a method to avoid this problem by providing a precise reduction from k-sampling without replacement to k-sampling with replacement. We call the resulting method Cascade Sampling.
\end{abstract}

\section{Introduction}

\emph{Random sampling} is a fundamental tool that has many applications in
computer science (see e.g., Motwani and Raghavan \cite{randomized}, Knuth \cite{knuth}, Tille \cite{tille}, and Olken \cite{olken}). Random sampling methods are widely used is data stream processing
because of their simplicity and efficiency \cite{dlt,stream_operator,mg,dynamic,reservoir2,reservoir1}. In a stream, the size of the domain and the probability of sampling an element both change constantly; this makes the process of sampling non-trivial. We distinguish between
sampling \emph{with replacement}, where all samples are independent
(and thus can be repeated), and sampling
\emph{without replacement}, where repetitions are prohibited.

In particular, weighted sampling without replacement has proven to be a very important tool. In weighted sampling,  each element is given a weight, where the probability of an element being selected is based on its weight. 
In their work Efraimidis and Spirakis \cite{weighted_power} presented an algorithm for weighted sampling
without replacement. Cohen and Kaplan \cite{Cohen:2008:TEU:1453856.1453884} use similar methods for their bottom-k
sketches. While their preliminary implementation yielded promising results, Efraimidis and Spirakis \cite{weighted_power} state, as the main open problem of the paper, \emph{``However, the question if, and to what extent, the finite precision arithmetic affects the algorithms remains an open problem.''}

In this paper we continue this work and provide a new algorithm to avoid the issue of relying on finite precision arithmetic. With this result we show that precision loss is not required in order to sample without replacement. We accomplish this by providing a precise reduction from $k$-sampling without replacement to $k$-sampling with replacement, using a special case of $k$-sampling with replacement, unit sampling (where $k$=$1$).  Additionally, we believe that in the future our method of expressing different random samples via reduction will provide a tool that allows further translation of other sampling methods into a more effective form for streams.

\subsection{Related Work}
Due to its fundamental nature, the problem of random sampling has received considerable attention
in the last few decades.

In 2005, Vitter \cite{reservoir} presented uniform sampling using a reservoir (with
and without replacement) over streams.
Further, the question of reductions between sampling methods has been addressed before.
For instance, Chaudhuri, Motwani and Narasayya \cite{joins} briefly discuss reductions for various sampling methods.
Cohen and Kaplan \cite{Cohen:2008:TEU:1453856.1453884} use a ``mimicking process'' in their papers, which is essentially a reduction
from sampling without replacement to sampling with replacement.

Chaudhuri, Motwani and Narasayya \cite{joins} use
the well-known method of ``over-sampling'', i.e. we sample the set independently
until $k$ distinct elements are obtained. Clearly, this schema does not introduce any precision loss,
since unit sampling is used as a black-box.

Unfortunately, the amount of resources required to determine this information is a function of  the weight distribution for the data set, and thus can be arbitrarily large.

In particular, consider the case when there is an element with weight that is overwhelmingly larger than
the rest of the population. In this case, the number of repetitions found while sampling with replacement is significantly larger then $k$.

Probably the first effective non-streaming solution for the weighted sampling without replacement problem was the algorithm
of Wong and Easton \cite{weighted_non_stream}. It is used by many other algorithms (see Olken \cite{olken} for the discussion).
For data streams, Efraimidis and Spirakis \cite{weighted_power} proposed an algorithm that is based on the ``exponent method''.
The algorithm requires precise computations of random keys $r^{1/w(p)}$, where $r\sim U[0,1]$.
The sample generated is composed of the $k$ elements with maximal keys. Cohen and Kaplan \cite{Cohen:2008:TEU:1453856.1453884} used similar methods as a building block for their bottom-k
sketches. The bottom-k sketch is an effective construction that has been extensively used for various
applications including approximations of aggregative queries over data streams. As Cohen and Kaplan \cite{Cohen:2008:TEU:1453856.1453884} show,
these methods are very effective in practical applications and are superior to the sketches that are based on sampling with replacement.

While effective in practice, the algorithms of Efraimidis and Spirakis and Cohen and Kaplan
introduce a loss of accuracy, since their techniques require additional floating point arithmetic operations.

\subsection{Results}

In this paper we show that the tradeoff between precision and performance is not a necessary property of sampling without replacement from data streams and construct a precise streaming reduction from $k$-sampling without replacement to $k$-sampling with replacement. This result provides a practical improvement to the algorithms of Efraimidis and Spirakis in cases where high accuracy is required. 

Our method is yields a surprisingly simple algorithm, given the importance of sampling without replacement and the existence of many previous methods. We call this algorithm Cascade Sampling.
In particular, when used with the algorithm from \cite{joins} Cascade Sampling requires $O(k)$ memory, constant time per element and the same precision as in \cite{joins}. 

\subsection{Intuition}

Let $\Lambda$ be \emph{any} algorithm that maintains a unit weighted sample from stream $D$.
Similarly to the over-sampling method, we maintain instances of $\Lambda$. Namely, we maintain $k$ instances $\Lambda_1,\dots, \Lambda_k$. However, we introduce the idea of \emph{stream modification}. That is, instead of applying $\Lambda$ independently and symmetrically on $D$, we apply $\Lambda_i$ on the modified stream $D_i$ that does not contain samples of $\Lambda_j$ for $j<i$.
In particular, $\Lambda_i$ may process its input elements in an order \emph{different from the order of their arrival in $D$}.
This simple but novel idea is sufficient to solve the problem.
In particular, we can claim that the input of $\Lambda_i$ is a random set that precisely matches the definition of weighted sampling without replacement. Since we use $\Lambda$ as a black box with only a constant number of auxiliary variables, specifically pointers,
the resulting schema is a precise reduction.

\section{Definitions}

An important building block of our algorithm is the concept of a unit sample, that is, the ability to sample a single element from a set. 

\begin{definition}
Let $S$ be a finite set of elements and let $w$ be a non negative function $w: S \rightarrow R$.
A random element $X_S$ with values from $S$ is a \emph{\textbf{unit weighted random sample}} if, for any $a\in S$,
$P(X_S = a) = {w(a) \over w(S)}$. Here $w(S) = \sum_{a\in S} w(a)$.
\end{definition}

For an algorithm instantiating weighted unit sampling we provide Black-Box WR2 from \cite{joins}. Black-Box WR2 is a unit sample when $r=1$. 

\begin{algorithm}[H]\label{j}
\caption{Black-Box WR2: Algorithm for Weighted Unit Sampling}

\begin{enumerate}
\item $W \gets 0$.
\item Initialize reservoir with length $r=1$, $\lambda_0$.

\item For each tuple $t$ in stream:
\begin{enumerate}
\item Get next tuple $t$ with weight $w(t)$
\item $W \gets W + w(t)$
\item Set $\lambda_0 = t$ with prob. ${w(t) \over W}$
\end{enumerate}

\item Return $\lambda_0$
\end{enumerate}

\end{algorithm}

\begin{definition}
A \emph{\textbf{data stream}}  is an ordered, set of
elements, $p_1, p_2, \dots, p_n$, that can be observed only
once. An algorithm $A$ is a streaming sampling algorithm if $A$ outputs a sample using a single pass over the data set.
\end{definition}
\begin{definition}
A set $X = \{X_1,\dots, X_k\}$ is
called a \emph{\textbf{k-sample with replacement}} from $S$ if $X_1,\dots, X_k$ are independent random unit samples from $S$.
\end{definition}

Another fundamental sampling method is \emph{weighted sampling without replacement}.
\begin{definition}
Let $S$ be a finite set such that $|S|\ge k$.
An ordered set $X = \{X_1,\dots, X_k\}$ is called a \emph{\textbf{$k$-sample without replacement}} from $S, |S| \ge k$
if $X_1$ is a weighted unit sample from $S$ and for any $j>1$, $X_j$ is a weighted unit sample from $S\setminus \{X_1,\dots,X_{j-1}\}$.
\end{definition}
\begin{definition}
We say that there exists an a \emph{reduction} from a $k$-sampling to a unit sampling if for
any unit sampling algorithm $\Lambda$
there exists a $k$-sampling algorithm $\Upsilon = \Upsilon(\Lambda)$ that uses $\Lambda$ as a black-box.
We say that the reduction is \emph{precise} if for any $\Lambda$ that requires memory $m$ and time $t$:
\begin{enumerate}
\item $\Upsilon(\Lambda)$ requires $O(km)$ memory and $O(kt)$ time.
\item $\Upsilon(\Lambda)$ only uses comparisons (in addition to using $A$ as a black box).
\end{enumerate}
In other words, $\Upsilon(\Lambda)$ does not introduce any precision loss.
\end{definition}
There exists a (trivial) precise reduction from weighted sampling with replacement to unit sampling.
In this paper we give the first precise streaming reduction for weighted sampling without replacement to unit sampling.

\section{Cascade Sampling}
Let $S$ be a finite set such that $|S|\ge k$ and let $a\notin S$. Denote $T= S\cup \{a\}$, and let $w: T\mapsto R^+$ be a function.
Let $\{X_1,\dots, X_k\}$ be a $k$-sample without replacement from $S$ with respect to $w$.
Define an ordered sequence $\{Y_1,\dots, Y_k\}$\footnote{Here the additional randomness is independent.} as follows:
\begin{equation}\label{sddsfsdfsdfdsdsd}
Y_{1} =
\left\{
  \begin{array}{ll}
    a, & \hbox{w.p. \ \ $w(a) \over w(T)$;} \\
    X_1, & \hbox{otherwise.} \\
  \end{array}
\right.
\end{equation}
For $i\ge 1$ define\footnote{Here $\setminus$ denotes the set difference, i.e. $A\setminus B = \{x: x\in A, x\notin B\}$.}:
\begin{equation}\label{sddsfsdf}
L_i = \{X_1,\dots, X_i, a\} \setminus \{Y_1,\dots, Y_i\}.
\end{equation}
We will show that $|L_i| = 1$; assuming that, let $Z_i$ be the single element from $L_i$, i.e., $L_i = \{Z_i\}$.
Put $U_i = T\setminus \{Y_1,\dots, Y_i\}.$ Define
\begin{equation}\label{sddsfsdfsdfdsdsd}
Y_{i+1} =
\left\{
  \begin{array}{ll}
    Z_{i}, & \hbox{w.p. \ \ $w(Z_{i}) \over w(U_i)$;} \\
    X_{i+1}, & \hbox{otherwise.} \\
  \end{array}
\right.
\end{equation}
\begin{lemma}\label{sddsfsdfsdf}
For all $i=1,\dots, k$ the ordered set $\{Y_1,\dots, Y_i\}$ is an  $i$-sample without replacement from $T$ with respect to $w$.
\end{lemma}
\begin{proof}
We prove the lemma by induction on $i$. For $i=1$ the statement follows from direct computation and definitions.
Assuming that the lemma is correct for $i$ we need to prove that
\begin{equation}\label{sdffdfd}
Y_{i+1} \in T\setminus \{Y_1,\dots, Y_i\},
\end{equation}
and for any $b\in U_i$:
\begin{equation}\label{ddsdfsdfsdf}
P(Y_{i+1} = b) = {w(b)\over w(U_i)}.
\end{equation}
To show $(\ref{sdffdfd})$ observe that $\{Y_1,\dots, Y_i\} \subseteq \{X_1, \dots, X_i, a\}$ and $Y_{i+1} \in \{X_{i+1}, Z_{i}\}$.
By definition $X_{i+1} \notin \{X_1, \dots, X_i, a\}$  and $Z_{i} \notin \{Y_1,\dots, Y_i\}$.

To show $(\ref{ddsdfsdfsdf})$ fix $\{X_1,\dots, X_i\}$ and $\{Y_1,\dots, Y_i\}$; it follows that $Z_i$ is fixed as well.
Denote $V_i = U_i \setminus \{Z_{i-1}\}$ and $H_i = S\setminus \{X_1,\dots, X_i\}$; it follows that $H_i = V_i$.
For any fixed $b\in V_i$ we have
$$
P(Y_{i+1} = b) = P(X_{i+1} = b){w(V_i) \over w(U_i)} = {w(b) \over w(H_i)}{w(V_i) \over w(U_i)} = {w(b) \over w(U_i)}.
$$
The case $b= Z_{i-1}$ is similar.
\end{proof}

\section{Precise Reduction and Resulting Algorithm}\label{sddssdsdfsdf}

Let $\Lambda$ be an algorithm that maintains a unit weighted sample from $D$. The algorithm from \cite{joins} is an example of $\Lambda$ but our reduction works with \emph{any} algorithm for unit weighted sampling.
We construct an algorithm $\Upsilon = \Upsilon(\Lambda)$ such that $\Upsilon$ maintains a $k$-sample without replacement.
Specifically, we maintain $k$ instances of $\Lambda$: $\Lambda_1, \dots, \Lambda_k$ such that the input of $\Lambda_i$ is a random substream of $D$ that is selected in a special way. We denote the input stream for $\Lambda_i$ as $D_i$.
Let $X_i$ be the sample produced by $\Lambda_i$.
The critical observation is that our algorithm maintains the following invariant: at any moment $D_i = D \setminus \{X_1, \dots, X_{i-1}\}$. Thus, by definition, the weighted sample from $D_i$ is the $i$-th weighted sample from $D$ when the samples are without replacement.

\begin{theorem}\label{lm:strem weighted}
Algorithm $\Upsilon = \Upsilon(\Lambda)$ maintains a weighted $k$-sample without replacement from $D$.
If $\Lambda$ requires space $O(g)$ and time per element $O(h)$, then $\Upsilon$ requires $O(kg)$ space and $O(kh)$ time respectfully.
Thus, there exists a precise reduction from $k$-sampling without replacement to a unit sampling.
\end{theorem}
\begin{proof}

Follows from the description of the algorithm (See Algorithm 2) and Lemma \ref{sddsfsdfsdf}.
\end{proof}

\begin{algorithm}[H]
\caption{Cascade Sampling}

\emph{Input: Data Stream $D = \{p_1,\dots, p_n\}$},\\
\ \ \ \ \emph{$\Lambda$ is an algorithm that maintains a unit weighted sample from $D$},\\
\ \ \ \ \emph{$\Lambda_1,\dots, \Lambda_k$ are independent instances of $\Lambda$}\\
\emph{Output: Weighted $k$-Sample Without Replacement $\{Y_1,\dots, Y_k\}$ }
\begin{enumerate}
  \item For $j= 1,2,\dots, n$
  \begin{enumerate}
  \item $new = p_j$
  \item For $i = 1, \dots, \min\{j,k\}$
  \begin{enumerate}
  \item If ($i<j$) then set $previous = Y_i$  (where $Y_i$ the current output of $\Lambda_i$).
  \item Feed $\Lambda_i$ with $new$
  \item If $Y_i$ changes its value to $new$, then set $new = previous$.
  \end{enumerate}
  \end{enumerate}
  \item Output $\{Y_1,\dots, Y_k\}$

\end{enumerate}

\end{algorithm}

Algorithm 2 provides a solution to the weighted $k$-Sampling without replacement problem. To better demonstrate the algorithm, we show an example of updating a single unit sample inside of loop \emph{(b)} in Figure 1. In this example, unit sample $\lambda_1$ has currently sampled element $a$ and unit sample $\lambda_2$ has currently sampled element $b$, where $a$ and $b$ are elements that appeared previously in the stream.

\begin{figure}[h]\label{pic1}
\centering
\includegraphics[scale=0.5]{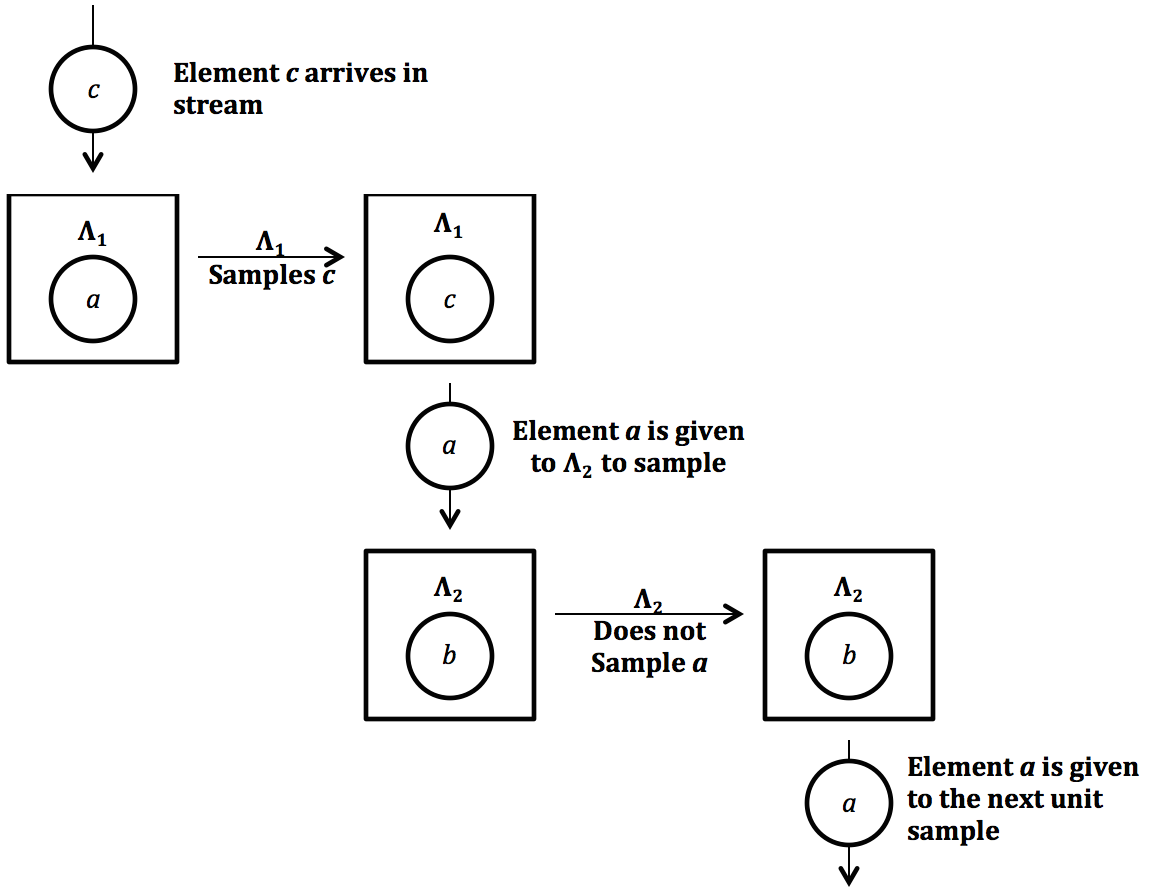}\caption{Updating a Unit Sample}
\setlength{\belowcaptionskip}{-10pt}
\clearpage
\end{figure}

\subsection{Discussion}
There are several directions in which our algorithm can be improved. In particular, run time dependent on the number of samples is one issue for practical datasets with large k. We believe this can be improved by combining several sampling steps into a single step which will be useful for the cases when the element will not be sampled into any of the substreams. This will often be the cases with elements with small weights. Specifically, we ask if it is possible to reduce the total running time from $O(nk)$ to $O(n \log k)$.

Another interesting direction is applying this algorithm to weighted random sampling with a bounded number of replacements as shown in \cite{efraimidis2010weighted}. Finally, this method may also be interesting when applied to the Sliding Window Model \cite{our} and Streams with Deletions \cite{dynamic}.

We thank our anonymous reviewers for their helpful suggestions, particularly for suggesting interesting open problems for discussion.

\newpage

\bibliography{Bibliography}

\end{document}